\documentclass[12pt]{elsarticle} %[reqno]{amsart}
\usepackage{graphicx}
\usepackage{amsmath}
\usepackage{amssymb}
\usepackage{amsthm}
\usepackage{amsfonts}
\usepackage{epsfig}
\usepackage{bm}
\usepackage{dsfont}
\newtheorem{Theorem}{Theorem}[section]
\newtheorem{Lemma}{Lemma}[section]

\def\ee{\end{eqnarray*}}
\def\be{\begin{eqnarray*}}
\def\bee{\end{eqnarray}}
\def\bbe{\begin{eqnarray}}
\def\ea{\end{align*}}
\def\ba{\begin{align*}}
\def\baa{\end{align}}
\def\bba{\begin{align}}
\def\u{\bm u}
\def\f{\bm f}

    \def\p{\partial}

    \def\i{\mathrm i}

\makeatletter
\let\today\relax
\def\ps@pprintTitle{%
    \let\@oddhead\@empty
    \let\@evenhead\@empty
    \def\@oddfoot{\footnotesize\itshape
      \hfill\today}%   {Preprint submitted} \hfill\today}%
    \let\@evenfoot\@oddfoot
    }
\makeatother

\begin{document}

\title{Hopf bifurcation of a non-parallel Navier-Stokes flow }

\author{Zhi-Min Chen }%\corref{cor1} }

 \ead{zmchen@szu.edu.cn}
\address{School  of Mathematics and Statistics, Shenzhen University, Shenzhen 518060,  China}%

%\date{}% It is always \today, today,

\begin{abstract}

A plane non-parallel flow in a square fluid domain exhibits an odd number of vortices.  A spectral  structure is found to have a non-real  solution  of the spectral problem linearized around the flow. With the use of this structure, Hopf bifurcation or  secondary time periodic flows branching of a basic square eddy flow is found. In contrast to a square eddy flow involving an even number of vortices in earlier analytical and experimental investigations, instability of the flow leads to  steady-state bifurcations.

\begin{keyword} Hopf bifurcation, non-parallel flows, Navier-Stokes equations,  square eddy flows, bifurcating time periodic flows

 {\it 2020 MSC:} 35B32, 35Q30,76D05,76E09,76E25
 \end{keyword}
\end{abstract}

\maketitle

\section{Introduction}

 Transitions of a basic flow into steady-state flows and time periodic flows are inherent nonlinear phenomena of  viscous incompressible fluid motions governed by  Navier-Stokes equations. The study  of these phenomena is  the first stage for the understanding of the instability nature of     Navier-Stokes  flows.

 Some explicit sinusoidal Navier-Stokes models  have  been found to have the secondary time periodic flows.  For example,   Chen and Price \cite{Chen1999} and Chen et al. \cite{Chen2003} employed Fourier expansion showing the existence of time periodic bifurcating flows form  the parallel  Kolmogorov flow \cite{AM},  whose physical motion and instability were measured by Bondarenko {\it et al.} \cite{Bon79} in a laboratory experiment by taking a Hartmann layer friction into consideration.
%The Kolmogorov flow is a parallel flow.    Little is known about the  existence of a non-parallel  Navier-Stokes flow, which loses stability and undergoes Hopf bifurcation. %has  time periodic flows branching off a basic non-parallel flow.

In the present study,  we  consider  the existence of time periodic flows bifurcating from   the non-parallel square eddy  flow
 \be (-\sin x\cos y, \cos x\sin y)
   \ee
in the fluid domain  $\Omega =(0,N\pi)\times (0,N\pi)$ for an odd integer  $N\ge 1$.   When $N=1$, the flow is globally nonlinear stable \cite{Chen2019}. This rules out the existence of secondary flows. When the integer $N=3$, the occurrence of secondary time periodic flows were confirmed by numerical computations \cite{ChenPrice23}. When $n\ge 5$,  we shall  show the existence of time periodic flows branching off the basic square eddy flow in the direction of its critical spectral solutions.

Consider the two-dimensional fluid motion governed by the incompressible Navier-Stokes equations
\bbe\label{o1}
\frac{\p \u }{\p t}+ \u \cdot \nabla \u -\nu \Delta \u +\mu \u +\nabla p= \f,\,\,\,\nabla \cdot \u=0.
\bee
Here  $\nu>0$ is viscosity,  $\mu$ is an energy dissipation coefficient with respect to a   horizonal friction effect,      $\u$ is the  velocity and $p$ is the pressure of the fluid.
The vorticity formulation of (\ref{o1}) driven by the vortical forcing
$$ \f = (2\nu +\mu)( -\sin x\cos y, \cos x \sin y)$$
 can be expressed as (See \cite{Chen2021,Thess})
\bbe\label{oo1}
-\frac{\p \Delta \psi}{\p t }+ J(\psi,\Delta \psi)+ \nu \Delta^2 \psi-\mu \Delta \psi =(4\nu+2\mu)\sin x \sin y
\bee
for  the Jacobian $J(\varphi,\phi)=\p_y(\phi\p_x \varphi)-\p_x(\phi\p_y\varphi)$. Here   $\psi$ is the stream function subject to the equation
\be \u =\nabla\times \psi= \left(\frac{\partial \psi}{\partial y},-\frac{\p \psi}{\p x}\right).
\ee
 The boundary condition of the fluid motion is
\bbe \label{oo2} \psi|_{\p\Omega} =\Delta \psi|_{\p\Omega} =0.
\bee
The basic steady-state solution  of (\ref{oo1})-(\ref{oo2}) is the square eddy  flow
\be \label{oo4}\psi_0 = \sin x\sin y\ee
in the stream function formulation.

By taking the horizonal bottom Hartmann layer friction into consideration,
this  two-dimensional problem  is  approximated by a  thin layer three-dimensional magnetohydrodynamical fluid motion problem via the laboratory experiments of  Sommeria and  Verron \cite{Sommeria84} and Sommeria \cite{Sommeria86}, showing the presence of  secondary steady-state bifurcation flows. The equivalent formulation  of the experimental fluid motion model \cite{Sommeria86,Sommeria84} and (\ref{oo1})-(\ref{oo2})  is shown in  \cite{Chen2019}. The steady-state bifurcation theory in relating to the experimental steady-state bifurcation phenomena \cite{Sommeria86,Sommeria84}   for $N=2$  was obtained  by the author \cite{Chen2021}. The existence of  secondary steady-state flows in  rectangular  domains in response to the experimental secondary steady-state flows of  Tabeling {\it et al.} \cite{T1,T2} was presented recently by the author \cite{Chen2022}.

However, for the square eddy flow with an odd number of vortices, the instability of the square eddy flow behaviours in a different manner,  due to the presence of secondary time periodic flows.    A spectral structure with respect to this flow  is found  to have a complex spectral behaviour leading to the existence of  the secondary time periodic flows.
%The complex spectral analysis of (\ref{oo1})-(\ref{oo2}) in relating to  Hopf bifurcation have not been studied in earlier investigations.

 Linearizing (\ref{oo1})-(\ref{oo2}) by employing the perturbation
\be \psi =\psi_0 + e^{\lambda t} \psi'\ee
and omitting the superscript primes, we have the spectral problem
\bbe\label{spp} 0= -\lambda \Delta \psi +\nu \Delta^2 \psi -\mu \Delta \psi + L\psi\,\,\mbox{ for }\,\,\, L\psi= J( \psi_0, (\Delta +2)\psi),\bee
where the eigenfunction $\psi$ has  the Fourier expansion
 \bbe \psi = \sum_{n,m\ge 1} a_{n,m}\sin \frac{ nx}N\sin\frac{my}N\ne 0.\label{sppp}
\bee
Equivalently, the expansion coefficients $a_{n,m}$ are subject to  the algebraic equation \cite{Chen2019}
\begin{align}
0
=&\sum_{n,m\ge -N} \sin \frac{nx}{N}\sin \frac{my}{N}\Big\{4N (\lambda\beta_{n,m}+\mu\beta_{n,m} +\nu\beta^2_{n,m})a_{n,m}\nonumber
\\
&+(n-m)[(\beta_{n-N,m-N}-2)a_{n-N,m-N} \nonumber
-(\beta_{n+N,m+N}-2)a_{n+N,m+N}]\nonumber
\\&+(n+m)[(\beta_{n-N,m+N}-2)a_{n-N,m+N}
- (\beta_{n+N,m-N}-2)a_{n+N,m-N} ]\Big\}\label{alg}
\end{align}
for $\beta_{n,m}= \frac{n^2+m^2}{N^2}$ and  $a_{n,m}=0$ whenever $n\le 0$ or $m\le 0$.

Due to the non-parallel flow property of $\psi_0$ running in $x$ and $y$ directions, the  expansion coefficient $a_{n,m}$ is effected initially by the following four influence  coefficients
\be a_{n-N,m-N},\,\,a_{n-N,m+N},\,\,a_{n+N,m-N},\,\,a_{n+N,m+N}.\ee

When $N=2$, there is an  eigenfunction  balance between  $a_{n,m}$ and its initial influence  $a_{n\pm N,m\pm N}$ and $a_{n\pm N,m\mp N}$ for $(n,m)$ having the same  parity to $(n\pm N,m\pm N)$ and  $(n\pm N,m\mp N)$. Thus  a real and simple critical spectral solution was  found in a suitable function space, which is flow invariant  with respect to the associated nonlinear equation, and hence the existence of steady-state bifurcations  was obtained \cite{Chen2019,Chen2021}. This steady-state bifurcation analysis can be extended to an even number $N>2$.
 In contrast, for $N\ge 3$  odd,  there arises a different balance with   the indices  in the opposite parity.    This gives rise to the occurrence of   non-real critical eigenfunctions and  Hopf bifurcation, confirmed numerically \cite{ChenPrice23} at $N=3$. However,  this opposite parity property leads to the absence  of  a  nonlinear flow invariant subspace, in which the simplicity condition holds true.

As shown in \cite{Chen2019},  we may find  eigenfunctions as Fourier expansions   generated respectively by the modes $\sin \frac{n_0x}{N}\sin \frac{m_0y}{N}$ with
\bbe \beta_{n_0,m_0}=\frac{n_0^2+m_0^2}{N^2}<2. \label{n0m0}
\bee
To seek  a suitable eigenfunction space containing non-real critical eigenfunctions for $N\ge 3$ odd, we assume that the positive  index  $(n_0,m_0)$ is  subject to the condition
\bbe N=n_0+m_0 \,\,\, \mbox{ and } m_0=n_0+1. \label{NNN}\bee
Therefore,  the eigenfunction expansion coefficient $a_{n_0,m_0}$ has the initial influence coefficients
\bbe a_{-n_0+N, -m_0+N},  \,\,\, a_{-n_0+N, m_0+N}, \,\,\, a_{n_0+N, -m_0+N}, \,\,\, a_{n_0+N, m_0+N}
\bee
To check the validity of (\ref{n0m0}) for these indices, by (\ref{NNN}), we see that
\be
&&\beta_{-n_0+N,-m_0+N}<2,
\\
&&\beta_{-n_0+N,m_0+N} =\frac{(N+1)^2+(3N+1)^2}{4N^2}>2,
\\
&&\beta_{-m_0+N,n_0+N}= \frac{(N-1)^2+(3N-1)^2}{4N^2}>2, \mbox{ if } N\ge 5,
\\
&&\beta_{n_0+N,m_0+N}>2.
\ee
Hence, for $N>5$ shown in (\ref{NNN}), it follows from  (\ref{alg})  that  we have the eigenfunction $\psi$  generated exactly by the two modes
\bbe \label{ss}\sin \frac{n_0x}{N}\sin \frac{m_0y}{N} \,\,\mbox{ and }\,\, \sin \frac{m_0x}{N}\sin \frac{n_0y}{N},\bee
 with the integer vectors $(n_0, m_0)$  subject to (\ref{n0m0}).
Therefore,  for $N\ge 5$, the Fourier expansion of an eigenfunction generated by the modes (\ref{ss}) can be expressed  as
\bbe \psi &=&a_{n_0,n_0+1}\sin \frac{ n_0x}N\sin\frac{(n_0+1)y}N+a_{n_0+1,n_0}\sin \frac{ (n_0+1)x}N\sin\frac{n_0y}N\nonumber
 \\
 &&+ \sum_{n, m\ge 1,\, \beta_{n,m}>2} a_{n,m}\sin \frac{ nx}N\sin\frac{my}N  \label{sppp1}
\bee
for  the influence coefficients $a_{n,m}$  subject to the following possible form
\bbe  a_{n_0+jN,m_0+kN}, a_{m_0+jN,n_0+kN}, a_{n_0+jN,n_0+kN}, a_{m_0+jN,m_0+kN}\label{in} \bee
for  integers $j, k\ge 0$. It is readily seen that  $\beta_{n,m}\ne 2$ for the  indices $(n,m)$  displayed in (\ref{in}).

By a numerical computation,
  we obtained  the critical eigenfunction given by  (\ref{sppp1})  associated with  a pure imaginary  eigenvalue $\lambda\ne 0 $   and found that   the expansion coefficients  are subject to the condition
 \bbe
 a_{n,m} = -\i a_{m,n}, \,\, \mbox{ whenever  $n$ is even and $m$ is odd}.\label{spp10}
 \bee
 Note that  $n_0$ and $m_0=n_0+1$ have opposite parity. This leads to the opposite parity  of the indices $n$ and $m$ in (\ref{sppp1}). We thus  consider  eigenfunctions in the following  space
  \begin{align}
E_N= &\Big\{ \psi \in H^4 \, |\,\, \psi \mbox{ is expanded as in  (\ref{sppp1}), (\ref{NNN}) and (\ref{spp10}) so that the
}\nonumber
\\ &\mbox{  influence coefficients indices  $n$ and  $m$ have opposite parity } \Big\},\label{EN}
\end{align}  %= \Big(\sum_{n \,even,\, m\, odd}+\sum_{n \,odd,\, m\, even} \Big) a_{n,m} \sin\frac{nx}N \sin \frac{my}N ,
which is an invariant eigenfunction space of the spectral operator determined  on the right-hand side of (\ref{spp}).  Here
 $H^4$  is the Hilbert  space of functions  subject to the boundary condition (\ref{oo2}) and equipped with  the  norm $\|\Delta^2 \psi\|_{L_2(\Omega)}$.

To show the existence of periodic solutions of frequency $\omega$ of (\ref{oo1}),  we put $s=\omega t $ to rewrite  (\ref{oo1}) as
\bbe\label{aabb}
-\omega\frac{ \p  \psi}{\p s }+ \nu \Delta \psi-\mu \psi +\Delta^{-1}L\psi+  \Delta^{-1}J(\psi,\Delta \psi)=0.
\bee
Consider solutions  in the parabolic H\"older space $C^{2k+2\alpha,k+\alpha}$, the Banach space of functions subject to (\ref{oo2}) and equipped  with the H\"older norm $\|\cdot\|_{2k+2\alpha,k+\alpha}$ on the space time domain $\Omega \times [0,2\pi]$ (See \cite{Lady,Sat}).
Moreover, we assume
$C^{2k+2\alpha,k+\alpha}_{2\pi}$ the complete subspace of the functions   satisfying the $2\pi$-periodic condition in time.
 $L_2$ inner products  are in the following sense.
\be   (\phi,\varphi) = \frac4{N^2\pi^2}\int^{N\pi}_0\int^{N\pi}_0 \phi \bar\varphi dx dy \,\,\,\mbox{ and } \,\, \langle  \phi, \varphi\rangle  = \frac{1}{2\pi} \int^{2\pi}_0 (\phi,\varphi)ds.
\ee

   We are   now in the position to state the main result of the present study.

   \begin{Theorem} \label{th1} Let  $0<\alpha <\frac12$ and  $N\ge 5$ be an odd integer. Assume that the spectral problem  (\ref{spp})-(\ref{sppp})  has  a critical spectral solution
   \be (\lambda, \nu, \mu, \psi)=(\i \omega_c, \nu_c, \mu_c,  \psi_c) \,\,\mbox{ with }\,\, \nu_c>0,\, \mu_c \ge 0\ee
   with $\psi_c \in E_N$.
   Here  the real $\omega_c\ne 0$ is unique with respect to $(\nu_c,\mu_c)$.
Let $P$  be the projection operator   in the following sense
\bbe P \psi =\psi-  \langle \psi,e^{-\i s}\Delta\psi_c^*\rangle e^{-\i s} \psi_c-\langle \psi,e^{\i s}\Delta \overline{ \psi_c^*}\rangle  e^{\i s} \bar\psi_c\bee
where $\psi_c^*$ is the critical conjugate eigenfunction  defined by the conjugate eigenfunction problem
\bbe\label{cenn} 0= \i \omega_c  \Delta \psi_c^* +\nu_c \Delta^2 \psi_c^* -\mu_c \Delta\psi_c^* +L^*\psi^*_c
\bee
with $L^*$ defined through the $L_2$ conjugation
\bbe\label{psii} (L\phi,\varphi) =(\phi, L^*\varphi) \,\,\mbox{ for }\,\, \phi, \varphi \in H^4. \bee
If
\bbe \dim\Big\{ \psi \in H^4 \,|\, 0=( -\i \omega_c\Delta +\nu_c\Delta^2 -\mu_c\Delta + L) \psi \Big\} =1, \label{oneo}\bee
then for any $\epsilon\ne 0$ sufficiently small, equations (\ref{oo2}) and  (\ref{aabb}) admit   a unique  element   \be (\omega, \sigma, \phi) \in (-\infty,\infty)\times (-\infty, \infty)\times  PC^{2+2\alpha,1+\alpha}_{2\pi},\ee
 dependant smoothly on $\epsilon$, in a neighborhood of $(\omega_c,1, 0)$ so that
\be\psi=\epsilon (e^{-\i s} \psi_c+e^{\i s} \bar\psi_c + \phi) \,\,\mbox{ and } \,\, (\omega,\mu,\nu)=(\omega, \sigma \mu_c,\sigma \nu_c)
\ee  solve (\ref{oo2}) and (\ref{aabb}).

\end{Theorem}

In comparison with parallel Kolmogorov flow $\sin y$ running in the single $y$ direction, the Fourier expansion coefficient $b_n$ of the  corresponding eigenfunction  is effected initially  by two influence coefficients
denoted as $b_{n-N}$ and $b_{n+N}$. Therefore,  a continued fraction method \cite{S} was employed to show the existence of critical values leading to the occurrence  of Hopf bifurcation of Navier-Stokes flows \cite{Chen1999,Chen2003}. However, in the present non-parallel case, the existence of critical values remains open due to a  lack of suitable  rigorous analysis, although numerical critical values are available \cite{Chen2021,ChenPrice23}.

A  basic flow loses stability and gives rise  to a primary bifurcation, which   is essentially determined by  the spectral behaviour of a  linearized  equation around the  basic flow.
Steady-state bifurcation  motion arises   in the direction of a single real critical eigenfunction, while  Hopf bifurcation occurs in the direction of  a pair of conjugate critical eigenfunctions.
Therefore, in the viewpoint of rigorous analysis, the former is much easier to be observed in contrast  with the later \cite{K}.

In the pioneer work of Hopf \cite{H}, an eigenfunction simplicity condition and an eigenfunction transversal crossing condition was found to be  sufficient to ensure the existence of  Hopf bifurcation  for  a finite dimensional system.
The extension of this  result to  Navier-Stokes equations   was given by Joseph and Sattinger \cite{JS}.
However, as  emphasized by Kirchg\"assner \cite{K} that   it is very  difficult to find a concrete Navier-Stokes motion model displaying the validity of the  two spectral conditions.  Positive answers  to this validity was obtained in   \cite{Chen1999,Chen2003} for  parallel  Kolmogorov  related flows, whereas  there is little literature available  for  non-parallel flows. Amongst some assumptions,  Galdi \cite{G} studied a non-parallel flow  passing a cylinder and obtained   the existence of Hopf bifurcation by assuming  the transversal crossing condition.

For the present non-parallel square vortex flow, an attempt has been  made to show the validity of the transversal crossing condition in Section 3 by following the study of Kolmogorov flows \cite{Chen1999,Chen2003}.
 For the eigenfunction simplicity condition, we adopt the eigenfunction  space $E_N$, which is  invariant with respect to  the linear problem (\ref{spp}) and
\bbe \dim\Big\{ \psi \in E_N \,|\, 0=( -\i \omega_c\Delta +\nu_c\Delta^2 -\mu_c\Delta + L) \psi \Big\} \le 1. \label{oneoo}\bee
Indeed, assuming  the existence of two eigenfunctions $\psi,\,\psi'\in E_N$ with the  same coefficients $a_{n_0,n_0+1} = a'_{n_0,n_0+1}$ and taking the inner product of the critical spectral equation
\be 0=(-\i \omega_c\Delta +\nu_c\Delta^2 -\mu_c\Delta + L) (\psi-\psi') \ee
with $ (-\Delta -2) (\psi-\psi')$, we have
\be 0= \sum_{n,m\ge 1;\, \beta_{n,m} >2}( \nu_c \beta_{n,m}^2 +\mu_c \beta_{n,m} )(\beta_{n,m}-2) |a_{n,m}-a_{n,m}'|^2,
\ee
since $a_{n_0+1,n_0}=a'_{n_0+1,n_0}$ due to (\ref{spp10}). This gives $\psi=\psi'$ and the validity of (\ref{oneoo}). The equality of (\ref{oneoo}) can be shown by numerical computation. The eigenfunction space $E_N$ is not flow invariant with respect to the nonlinear problem (\ref{oo1}). We thus  assume the simplicity condition (\ref{oneo}) in the whole space $H^4$.

From the laboratory experiments of \cite{Sommeria84,Sommeria86} and  the numerical computation of \cite{Chen2021,ChenPrice23}, the first critical spectral solutions of (\ref{spp}) are always real and lead to steady-state bifurcations. The bifurcating steady-state solutions  are stable and thus observable experimentally \cite{Sommeria84,Sommeria86}. Thus the bifurcating time-periodic flows in the present study is not expected to be stable and is not observable experimentally.

This paper is  structured as follows.
In Section 2, we provide basic spectral properties for the Navier-Stokes equations for the use in the present study.  Based on  the results from Section 2,  we  derive  some spectral properties at a critical stage in Section 3.  Numerical computation of critical spectral  solutions are provided in Section 4, where we take $N=5$ as  displaying purpose.    A bifurcation theorem is  shown  in Section 5.  Theorem \ref{th1}  is obtained by the combination of the results in Sections 3 and 4. % As an example,  a pair of numerical critical eigenfunctions  for $N=5$ are presented   in Section 5.

\section{Preliminary  spectral analysis}

Let us begin with the eigenfunction energy identity, which is of fundamental  importance  in the present analysis.
\begin{Lemma}\label{ll1}
Let $(\lambda,\nu,\,\mu,\,\psi)\in \mathbb{C}\times [0,\infty)\times [0,\infty)\times H^4$ solve the spectral problem (\ref{spp})-(\ref{sppp}) with $N\ge 1$.
Then there holds the  identity
\bbe0
&=& \Re\sum_{n,m\ge 1} (\lambda \beta_{n,m} + \mu \beta_{n,m} +\nu \beta_{n,m}^2)( \beta_{n,m}-2) |a_{n,m}|^2. \label{id}
\bee
\end{Lemma}
\begin{proof}
Taking the inner product of the spectral equation (\ref{spp}) with $-(\Delta +2)\psi$ and then integrating by parts by the boundary condition (\ref{oo2}), we have
\be
0&=& \Re \left(-\lambda  \Delta \psi -\mu\Delta \psi+\nu \Delta^2 \psi+ J( \psi_0, (\Delta +2)\psi), -(\Delta +2)\psi\right)
\\
&=& \Re \left(-\lambda  \Delta \psi -\mu\Delta \psi+\nu \Delta^2 \psi, -(\Delta +2)\psi\right)
.\ee
This gives (\ref{id}) after an elemental calculation for the Fourier expansion of $\psi$.
\end{proof}

Consider  the conjugate spectral problem
\bbe\label{cen} 0= -\bar\lambda \Delta \psi^*-\mu\Delta\psi^* +\nu \Delta^2 \psi^* +L^*\psi^*.
\bee
  Here  $ L^*\psi^*= - (\Delta+2)J(\psi_0,\psi^*)$  is derived from (\ref{psii}).

We have the following  characterisation of the  conjugate spectral problem:
\begin{Lemma}\label{ll2} Let $(\lambda,\nu,\mu,\psi)\in \mathbb{C}\times (0,\infty)\times[0,\infty)\times  H^4$ solve the spectral problem  (\ref{spp}), (\ref{NNN}) and (\ref{sppp1}) with $N\ge 5$ odd.  Then  its   conjugate eigenfunction  presents in the following form
\bbe \psi^* =\sum_{n,m\ge 1} a_{n,m}^* \sin \frac{nx}N\sin \frac{my}N \,\,\mbox{ for  } a_{n,m}^* = (-1)^m (-\beta_{n,m}+2)\bar a_{n,m}\label{conj}
.\bee
\end{Lemma}
\begin{proof}
By the derivation of  the conjugate operator, observing that $\Delta+2$ is invertible on the set of all functions (\ref{sppp1}),  we can rewrite (\ref{cen}) as
\be 0&=&-\bar\lambda \Delta\psi^*-\mu\Delta \psi^*  +\nu \Delta^2\psi^*-(\Delta +2)J(\psi_0,\psi^*)
\\
&=&
 (\Delta+2)(-\bar\lambda \Delta-\mu\Delta   +\nu \Delta^2 -L)(\Delta +2)^{-1} \psi^*
\ee
or
\be 0= (-\bar\lambda  \Delta-\mu\Delta \!+\!\nu \Delta^2  \!-\! L)(\Delta \!+\!2)^{-1} \psi^*\,
\ee
for
\be (\Delta \!+\!2)^{-1} \psi^*\!=\hspace{-2mm}\sum_{n,m\ge 1}\hspace{-2mm} b_{n,m} \sin\frac{nx}N\sin \frac{my}N.
\ee
Hence, consulting  with (\ref{alg}), we have
\begin{align*}
0
=&\sum_{n,m\ge -N} \sin \frac{nx}{N}\sin \frac{my}{N}\Big\{4N(\bar\lambda \beta_{n,m}+\mu\beta_{n,m} + \nu \beta_{n,m}^2)b_{n,m}
\\
&-(n-m)[(\beta_{n-N,m-N}-2)b_{n-N,m-N}-(\beta_{n+N,m+N}-2)b_{n+N,m+N}]
\\
&-(n+m)[(\beta_{n-N,m+N}-2)b_{n-N,m+N}
- (\beta_{n+N,m-N}-2)b_{n+N,m-N} ]\Big\}.
\end{align*}
Since $m$ and $m\pm N$ have  opposite parity due to $N$ odd,  let  $b_{n,m} = (-1)^m \bar a_{n,m}$  and apply the complex conjugate to the resultant equation to transform   the previous equation into the following:
\begin{align*}
0
=&\hspace{-2mm}\sum_{n,m\ge -N} (-1)^m\sin \frac{nx}{N}\sin \frac{my}{N}\Big\{4N(\lambda \beta_{n,m} +\mu\beta_{n,m}+\nu \beta_{n,m}^2) a_{n,m}
\\
&+(n-m)[(\beta_{n-N,m-N}-2) a_{n-N,m-N}
 -(\beta_{n+N,m+N}-2) a_{n+N,m+N}]
\\&+(n+m)[(\beta_{n-N,m+N}-2) a_{n-N,m+N}
- (\beta_{n+N,m-N}-2) a_{n+N,m-N} ]\Big\},
\end{align*}
which is identical to the algebraic equation (\ref{alg}) for $\psi$. We thus have
$$a_{n,m}^*= (-1)^m (-\beta_{n,m}+   2) \bar a_{n,m}$$ and thus the validity of (\ref{conj}).
\end{proof}

\section{Verification of the transversal crossing condition}

For the spectral problem at the critical stage $(\omega_c,\nu_c,\mu_c)$, let $(\nu,\mu)=\sigma (\nu_c,\mu_c)$. Differentiating  (\ref{spp}) with respect to $\sigma$ at $\sigma=1$, we have
  \bbe
  \frac{\p\lambda}{\partial \sigma }  \Delta \psi =\nu_c \Delta^2\psi -\mu_c \Delta \psi-\i \omega_c\Delta \frac{\p \psi}{\p \sigma}  +\nu_c \Delta^2\frac{\p \psi}{\p \sigma} -\mu_c \Delta \frac{\p \psi}{\p \sigma}+L\frac{\p \psi}{\p \sigma}
  \bee
Taking the inner product of the previous equation with the conjugate eigenfunction $\psi^*$, we have
 \bbe
  \frac{\p\lambda}{\partial \sigma }  (\Delta \psi,\psi^*) =(\nu_c \Delta^2\psi -\mu_c \Delta \psi,\psi^*)
  \bee
This section is devoted to show the transversal crossing condition  $\Re\frac{\p\lambda}{\partial \sigma } \ne 0$ at the critical stage $(\nu, \mu)= (\nu_c,\mu_c)$.

\begin{Theorem}\label{th2} For $N\ge 5$ odd, $\lambda= \i \Im \lambda\ne 0$,  $\nu>0$, $\mu \ge 0$ and $0\ne \psi\in E_N$, assume that   $(\lambda, \nu,\mu,\psi)$   solves  the spectral problem   (\ref{spp}), (\ref{NNN}) and (\ref{sppp1}).
  Then we have the properties

  \bbe
(  \psi,\Delta\psi^*)\neq 0, \label{z1}
\bee

  \bbe \Re \frac{(-\mu\psi+\nu\Delta \psi,\Delta \psi^*)}{(\psi,\Delta\psi^*)}\ne 0, \label{zzz}
  \bee
  where $\psi^*$ is the conjugate counterpart of $\psi$.

\end{Theorem}
Equation (\ref{zzz}) is actually the eigenvalue transversal crossing condition. For the square vortex flow problem, Theorem \ref{th2} shows the validity of the condition.

\begin{proof}

To prove (\ref{z1}), we see that $\psi\in E_N$ is in the form of (\ref{sppp1}).
Therefore, by Lemma \ref{ll2}, (\ref{sppp1}) and the property  $a_{n,m}= -\i a_{m,n}$ for $n$ even  and $m$ odd due to $\psi \in E_N$,  we  have
\begin{align}
 ( \psi,\Delta\psi^*)&=-( \sum_{\text{$n$ even, $m$ odd}}+\sum_{\text{$n$ odd, $m$ even}}) \beta_{n,m} a_{n,m} \bar {a^*}_{n,m}\nonumber
  \\
 &=\sum_{\text{$n$ odd, $m$ even}} \beta_{n,m}(\beta_{n,m}-2)(a_{n,m}^2-a_{m,n}^2)\nonumber
  \\
 &=\sum_{\text{$n$ odd, $m$ even}}2 \beta_{n,m}(\beta_{n,m}-2)a_{n,m}^2\label{d22}
  \\
 &=2 \beta_{n_0,n_0+1}(\beta_{n_0,n_0+1}\!-\!2)a_{n_0,n_0+1}^2
 +\hspace{-9mm}\sum_{\text{$n$ odd, $m$ even}; \beta_{n,m}>2}\hspace{-10mm}2 \beta_{n,m}(\beta_{n,m}\!-\!2)a_{n,m}^2.\label{d11}
 \end{align}
By Lemma \ref{ll1} and the property  $\Re\lambda =0$,  we have
\bbe
0&=& \sum_{\text{$n$ odd, $m$ even}}\!\!\!\!\! 2(\mu\beta_{n,m}+\nu\beta^2_{n,m})(\beta_{n,m}-2)|a_{n,m}|^2\nonumber
\\
&=&2(\mu\beta_{n_0,n_0+1}+\nu\beta_{n_0,n_0+1}^2)(\beta_{n_0,n_0+1}-2)|a_{n_0,n_0+1}|^2  \nonumber
\\
&&+\sum_{\text{$n$ odd, $m$ even}, \beta_{n,m}>2}\!\!\!\!\! 2(\mu\beta_{n,m}+\nu\beta^2_{n,m})(\beta_{n,m}-2)|a_{n,m}|^2\label{idd}
.\bee
The combination of (\ref{d11}) and (\ref{idd}) implies
\begin{align}
 |( \psi,\Delta\psi^*)|
 &\ge  -2\beta_{n_0,n_0+1}(\beta_{n_0,n_0+1}\!-\!2)|a_{n_0,n_0+1}|^2 \nonumber
 \\
 &-\sum_{\text{$n$ odd, $m$ even}, \,\beta_{n,m}>2}2\beta_{n,m}(\beta_{n,m}-2)|a_{n,m}|^2\nonumber
 \\
&= \sum_{\text{$n$ odd, $m$ even},\, \beta_{n,m}>2}\!\!\!\!\! 2\frac{(\mu+\nu\beta_{n,m})\beta_{n,m}(\beta_{n,m}-2)|a_{n,m}|^2}{\mu +\nu \beta_{n_0,n_0+1}}\nonumber
 \\
 &-\sum_{\text{$n$ odd, $m$ even},\, \beta_{n,m}>2}2\beta_{n,m}(\beta_{n,m}-2)|a_{n,m}|^2\nonumber
 \\
 &= \!\!\!\!\!\!\!\!\!\!  \sum_{\text{$n$ odd, $m$ even}, \,\beta_{n,m}>2}\!\!\!\!\!\!\!\!\!\! 2\frac{ \nu(\beta_{n,m}-\beta_{n_0,n_0+1})\beta_{n,m}(\beta_{n,m}-2)|a_{n,m}|^2}{\mu+\nu\beta_{n_0,n_0+1}}>0\label{aa2}
.\end{align}
We thus obtain (\ref{z1}).

To prove (\ref{zzz}), we use Lemma \ref{ll2} and  the property  $a_{m,n}= -\i a_{n,m}$ for $m$ even  and $n$ odd  to produce
\bbe
 \lefteqn{( -\mu\psi+\nu\Delta \psi,\Delta\psi^*)}\nonumber\\
 &=& (\sum_{\text{$n$ even, $m$ odd}} +\sum_{\text{$n$ odd, $m$ even}})(\mu\beta_{n,m}+\nu\beta^2_{n,m}) a_{n,m} \overline{a_{n,m}^*}\nonumber
 \\
 &=&-\ (\sum_{\text{$n$ even, $m$ odd}} +\sum_{\text{$n$ odd, $m$ even}}) (-1)^m (\mu\beta_{n,m}+\nu\beta^2_{n,m}) (\beta_{n,m}-2)a_{n,m}^2\nonumber
 \\
 &=&-\sum_{\text{$n$ odd, $m$ even}} (\mu\beta_{n,m}+\nu\beta^2_{n,m})(\beta_{n,m}-2)(a_{n,m}^2-a_{m,n}^2)\nonumber
 \\
 &=&-\sum_{\text{$n$ odd, $m$ even}} 2(\mu\beta_{n,m}+\nu\beta^2_{n,m})(\beta_{n,m}-2)a_{n,m}^2.\label{d2}
  \bee
  By (\ref{sppp1}), we have
 \begin{align}
  (-\mu\psi +\nu\Delta \psi,\Delta\psi^*)&=(-1)^{n_0}2(\mu\beta_{n_0,n_0+1}+\nu\beta^2_{n_0,n_0+1})(\beta_{n_0,n_0+1}-2)a_{n_0,n_0+1}^2\nonumber
\\
 &-\hspace{-10mm}\sum_{\text{$n$ odd, $m$ even},\,\beta_{n,m}>2}\hspace{-10mm} 2(\mu\beta_{n,m}+\nu\beta^2_{n,m})(\beta_{n,m}-2)a_{n,m}^2.\label{d3}
\end{align}

On the other hand, we see that
$$
(-\mu\psi\!+\!\nu\Delta \psi,\Delta\psi^*)=(-\mu\!-\!\nu\beta_{n_0,n_0+1})(\psi,\Delta\psi^*)\!+\!\nu((\Delta+\beta_{n_0,n_0+1})\psi,\Delta\psi^*).
$$
By (\ref{z1}), we have
\bbe
\frac{(-\mu\psi +\nu\Delta \psi,\Delta\psi^*)}{(\psi,\Delta\psi^*)}=-\mu -\nu\beta_{n_0,n_0+1} +   \nu\frac{\Big((\Delta+\beta_{n_0,n_0+1})\psi,\Delta\psi^*\Big)}{(\psi,\Delta\psi^*)}.\label{non}
\bee

If we have the non-zero   imaginary part
\bbe
 \Im\frac{(-\mu\psi +\nu\Delta \psi,\Delta\psi^*)}{(\psi,\Delta\psi^*)}\ne 0,\label{aa5}
 \bee
  equation (\ref{non}) together with (\ref{d2}) and (\ref{aa5}) implies
\begin{align*}
\Re&\frac{(-\mu\psi +\nu\Delta \psi,\Delta\psi^*)}{(\psi,\Delta\psi^*)}
\\
&=-\mu -\nu\beta_{n_0,n_0+1}+ \nu\Re\frac{\Big(( \Delta+\beta_{n_0,n_0+1})\Delta\psi, \psi^*\Big)}{(\psi,\Delta\psi^*)}
\\
&<-\mu -\nu\beta_{n_0,n_0+1}+ \nu\Bigg|\frac{\Big(( \Delta+\beta_{n_0,n_0+1})\Delta\psi, \psi^*\Big)}{(\psi,\Delta\psi^*)}\Bigg|
\\
&=-\mu -\nu\beta_{n_0,n_0+1} + \frac{\Big|\displaystyle\sum_{\text{$n$ odd, $m$ even},\,\beta_{n,m}>2}\hspace{-5mm}
 2\nu( \beta_{n,m}\!-\!\beta_{n_0,n_0+1})\beta_{n,m}(\beta_{n,m}\!-\!2)a_{n,m}^2\Big|}{|(\psi,\Delta\psi^*)|}
\\
&\le-\mu -\nu\beta_{n_0,n_0+1} + \frac{\displaystyle\sum_{\text{$n$ odd, $m$ even},\,\beta_{n,m}>2}\!\!\!\!\!\!\!\!
 2\nu( \beta_{n,m}\!-\!\beta_{n_0,n_0+1})\beta_{n,m}(\beta_{n,m}\!-\!2)|a_{n,m}|^2}{|(\psi,\Delta\psi^*)|}
 \\
\!\!&=-\mu -\nu\beta_{n_0,n_0+1} + \frac{\displaystyle\sum_{\text{$n$ odd, $m$ even}, \,\beta_{n,m}>2} 2\nu\beta_{n,m}^2(\beta_{n,m}-2)|a_{n,m}|^2}{|(\psi,\Delta\psi^*)|}
\\
&-\frac{\displaystyle\sum_{\text{$n$ odd, $m$ even}, \,\beta_{n,m}>2} 2\nu\beta_{n_0,n_0+1}\beta_{n,m}(\beta_{n,m}-2)|a_{n,m}|^2}{|(\psi,\Delta\psi^*)|},
\end{align*}
which equals,
\begin{align*}
\!\!&-\mu -\nu\beta_{n_0,n_0+1}  + \frac{\displaystyle\sum_{\text{$n$ odd, $m$ even}, \,\beta_{n,m}>2} 2(\mu+\nu\beta_{n,m})\beta_{n,m}(\beta_{n,m}-2)|a_{n,m}|^2}{|(\psi,\Delta\psi^*)|}
\\
\!\!&\!\!-\frac{\displaystyle\sum_{\text{$n$ odd, $m$ even},\, \beta_{n,m}>2} 2(\mu+\nu\beta_{n_0,n_0+1})\beta_{n,m}(\beta_{n,m}-2)|a_{n,m}|^2}{|(\psi,\Delta\psi^*)|},
\end{align*}
which becomes, by  (\ref{sppp1}) and (\ref{idd}),
\be
&&  -\mu -\nu\beta_{n_0,n_0+1} - \frac{2(\mu+\nu\beta_{n_0,n_0+1})\beta_{n_0,n_0+1}(\beta_{n_0,n_0+1}-2)|a_{n_0,n_0+1}|^2}{|(\psi,\Delta\psi^*)|}
\\
&&- \frac{\displaystyle\sum_{\text{$n$ odd, $m$ even},\, \beta_{n,m}>2} 2(\mu+\nu\beta_{n_0,n_0+1})\beta_{n,m}(\beta_{n,m}-2)|a_{n,m}|^2}{|(\psi,\Delta\psi^*)|}
.\ee
Hence, by (\ref{d22}) and the condition $\beta_{n_0,n_0+1}<2$, we have
\be
\lefteqn{\Re\frac{(-\mu\psi +\nu\Delta \psi,\Delta\psi^*)}{(\psi,\Delta\psi^*)}}
 \\&<&  -\frac{\Big|\displaystyle\sum_{\text{$n$ odd, $m$ even}} 2(\mu+\nu\beta_{n_0,n_0+1})\beta_{n,m}(\beta_{n,m}-2)a_{n,m}^2\Big|}{|(\psi,\Delta\psi^*)|}
\\
&&+\frac{2(\mu+\nu\beta_{n_0,n_0+1})\beta_{n_0,n_0+1}|\beta_{n_0,n_0+1}-2||a_{n_0,n_0+1}|^2}{|(\psi,\Delta\psi^*)|}
\\
&&- \frac{\displaystyle\sum_{\text{$n$ odd, $m$ even}, \,\beta_{n,m}>2} 2(\mu+\nu\beta_{n_0,n_0+1})\beta_{n,m}(\beta_{n,m}-2)|a_{n,m}|^2}{|(\psi,\Delta\psi^*)|},
\ee
which is upper bounded by
\be
&& \frac{|\displaystyle\sum_{\text{$n$ odd, $m$ even},\, \beta_{n,m}>2} 2(\mu+\nu\beta_{n_0,n_0+1})\beta_{n,m}(\beta_{n,m}-2)a_{n,m}^2|}{|(\psi,\Delta\psi^*)|}
\\
&&-\frac{\displaystyle\sum_{\text{$n$ odd, $m$ even}, \,\beta_{n,m}>2} 2(\mu+\nu\beta_{n_0,n_0+1})\beta_{n,m}(\beta_{n,m}-2)|a_{n,m}|^2}{|(\psi,\Delta\psi^*)|}\le 0
.\ee
This gives the validity of (\ref{zzz}) under  condition (\ref{aa5}).

If (\ref{aa5}) is not true, we have
\be
\Re \frac{(-\mu\psi +\nu\Delta \psi,\Delta\psi^*)}{(\Delta \psi,\psi^*)}= \frac{(-\mu\psi +\nu\Delta \psi,\Delta\psi^*)}{(\Delta \psi,\psi^*)}.
\ee
Thus it remains to show that
\be (-\mu\psi +\nu\Delta \psi,\Delta\psi^*)\ne 0.
\ee

Indeed, without loss of generality, we may assume $a_{n_0,n_0+1}=1$  as the spectral problem is linear. It follows from  (\ref{idd}) and (\ref{d3}) that
\begin{align*}
(-\mu\psi &+\nu\Delta \psi,\Delta\psi^*)
\\
&= (-1)^{n_0}2(\mu+\nu\beta_{n_0,n_0+1})\beta_{n_0,n_0+1}(\beta_{n_0,n_0+1}-2)
\\
&-\sum_{\text{$n$ odd, $m$ even},\, \beta_{n,m}>2} 2(\mu+\nu\beta_{n,m})\beta_{n,m}(\beta_{n,m}-2)a_{n,m}^2
\\
&= \hspace{-3mm}\sum_{\text{$n$ odd, $m$ even}, \,\beta_{n,m}>2} \hspace{-3mm} 2(\mu+\nu\beta_{n,m})\beta_{n,m}(\beta_{n,m}-2)((-1)^{n_0}|a_{n,m}|^2-a_{n,m}^2),
\end{align*}
which is not equal to zero, unless $\Im a_{n,m}\equiv 0$.
We thus obtain (\ref{zzz}).

The proof of Theorem  \ref{th2} is complete.
\end{proof}

\section{Numerical spectral solutions}

For $N=n_0+m_0\ge 5$ with  $m_0=n_0+1$, we consider the spectral problem (\ref{spp}) with the eigenfunction in the Fourier expansion  (\ref{sppp1}) generated by the principle  modes $\sin \frac{n_0x}{N}\sin \frac{m_0y}N$ and $\sin \frac{m_0x}{N}\sin \frac{n_0y}N$.
This spectral problem  specified in (\ref{alg}) can rewritten as
\begin{align}
0 \label{nspp}
=&\sum_{n,m\ge -N} \sin \frac{nx}{N}\sin \frac{my}{N}\Big\{ (\lambda+\mu +\nu\beta_{n,m})a_{n,m}
\\
&+\frac{n-m}{4N\beta_{n,m}}[(\beta_{n-N,m-N}-2)a_{n-N,m-N}
-(\beta_{n+N,m+N}-2)a_{n+N,m+N}]\nonumber
\\&+\frac{n+m}{4N\beta_{n,m}}[(\beta_{n-N,m+N}-2)a_{n-N,m+N}
- (\beta_{n+N,m-N}-2)a_{n+N,m-N} ]\Big\},\nonumber
\end{align}
This is a  coupled system of infinite algebraic equations. For  the understanding of the system, these individual equations are displayed as
\begin{align*}
0=&(\lambda+\mu +\nu\beta_{n_0,m_0})a_{n_0,m_0}
\\
&-\frac{1}{4N\beta_{n_0,m_0}}[(\beta_{m_0,n_0}-2)a_{m_0,n_0}-(\beta_{n_0+N,m_0+N}-2)a_{n_0+N,m_0+N}]
\\
&-\frac{N}{4N\beta_{n_0,m_0}}[(\beta_{m_0,m_0+N}-2)a_{m_0,m_0+N}-(\beta_{n_0+N,n_0}-2)a_{n_0+N,n_0}
],
\\
0=&(\lambda+\mu +\nu\beta_{m_0,n_0})a_{m_0,n_0}
\\
&+\frac{1}{4N\beta_{m_0,n_0}}[(\beta_{n_0,m_0}-2)a_{n_0,m_0}-(\beta_{m_0+N,n_0+N}-2)a_{m_0+N,n_0+N}]
\\
&+\frac{N}{4N\beta_{m_0,n_0}}[(\beta_{m_0+N,m_0}-2)a_{m_0+N,m_0}-(\beta_{n_0,n_0+N}-2)a_{n_0,n_0+N}
],
\\
%\end{align*}\begin{align}
0
=&(\lambda+\mu +\nu\beta_{n_0+N,n_0})a_{n_0+N,n_0}
\\
&+\frac{2n_0+N}{4N\beta_{n_0+N,n_0}}[(\beta_{n_0+2N,m_0}-2)a_{n_0+2N,m_0}+(\beta_{n_0,n_0+N}-2)a_{n_0,n_0+N}]\nonumber
\\
&-\frac{N}{4N\beta_{n_0+N,n_0}}[(\beta_{n_0,m_0}-2)a_{n_0,m_0}+(\beta_{n_0+2N,n_0+N}-2)a_{n_0+2N,n_0+N}],\nonumber
\\ %\end{align}\begin{align*}
0
=& (\lambda+\mu +\nu\beta_{n_0,n_0+N})a_{n_0,n_0+N}
\\
&-\frac{2n_0+N}{4N\beta_{n_0,n_0+N}}[(\beta_{m_0,n_0+2N}-2)a_{m_0,n_0+2N}+ (\beta_{n_0+N,n_0}-2)a_{n_0+N,n_0}  ]\nonumber
\\
&+\frac{N}{4N\beta_{n_0,n_0+N}}[(\beta_{m_0,n_0}-2)a_{m_0,n_0}+ (\beta_{n_0+N,n_0+2N}-2)a_{n_0+N,n_0+2N}],\nonumber
\\
& ...,
\end{align*}\vspace{-8mm}
\begin{align*}
0
=&(\lambda+\mu +\nu\beta_{n,m})a_{n,m}
\\
&+\frac{n-m}{4N\beta_{n,m}}[(\beta_{n-N,m-N}-2)a_{n-N,m-N}
-(\beta_{n+N,m+N}-2)a_{n+N,m+N}]\nonumber
\\&+\frac{n+m}{4N\beta_{n,m}}[(\beta_{n-N,m+N}-2)a_{n-N,m+N}
- (\beta_{n+N,m-N}-2)a_{n+N,m-N} ]\Big\},\nonumber
\\
&...
\end{align*}
for $(n,m)= (n_0+jN,m_0+kN),$ $ (m_0+jN,n_0+kN),$ $(n_0+jN,n_0+kN)$ and $(m_0+jN,m_0+kN)$, respectively,  with positive integers $j$ and $k$.

For the odd integer $N\ge 5$, this equation system is essentially the same from viewpoint of  numerical computation. Numerical critical spectral solution
%\be (\lambda,\nu,\mu,\psi)= (\i\omega_c,\nu_c,\mu_c,\psi_c)\ee
is obtained.
 The numerical  eigenvalue  $\lambda=\i \omega_c \ne 0$ reaches the imaginary line at a   critical value  $(\nu_c,\mu_c)$.
This critical parameter $(\omega_c,\nu_c,\mu_c)$ is unique under the condition $\mu_c=0$.
The corresponding numerical  eigenfunction $\psi_c$  is unique up to a constant factor. The coefficients of $\psi_c$ are subject to the condition
\bbe a_{n,m} =-\i a_{m,n},\,\,\, \mbox{ for  $n$ even and $m$ odd.}\label{anm}
\bee
Hence the equality of (\ref{oneoo}) is true by numerical computation.

%This critical parameter $(\omega_c,\nu_c,\mu_c)$ is unique and is different to those of the  eigenfunctions in $H^4$ generated by other principle modes. That is, the simplicity condition (\ref{oneo}) is valid from viewpoint of numerical computation.

 For displaying purpose, we only provide computation results for $N=5$. The corresponding spectral equations become
 \begin{align*}
%\hspace{-20mm}
%\left.\begin{array}{ll}
0=& (\lambda\!+\!\mu \!+\!\nu\beta_{2,3})a_{2,3}
\\
&
-\frac{(\beta_{3,2}-2)a_{3,2}-(\beta_{7,8}-2)a_{7,8}}{20\beta_{2,3}}
-\frac{(\beta_{3,8}-2)a_{3,8}-(\beta_{7,2}-2)a_{7,2}}{4\beta_{2,3}},%\vspace{1.5mm}
\\
0=& (\lambda\!+\!\mu \!+\!\nu\beta_{3,2})a_{3,2}
\\
&
+\frac{(\beta_{2,3}-2)a_{2,3}-(\beta_{8,7}-2)a_{8,7}}{20\beta_{3,2}}
+\frac{(\beta_{8,3}-2)a_{8,3}-(\beta_{2,7}-2)a_{2,7}}{4\beta_{3,2}},%\vspace{1.5mm}
\\
0
=&(\lambda\!+\!\mu \!+\!\nu\beta_{7,2})a_{7,2}
\\
&
+9\frac{(\beta_{12,3}-2)a_{12,3}+(\beta_{2,7}-2)a_{2,7}}{20\beta_{7,2}}\nonumber
-\frac{(\beta_{2,3}-2)a_{2,3}+(\beta_{12,7}-2)a_{12,7}}{4\beta_{7,2}},%\vspace{1.5mm}
\\ %\end{align}\begin{align*}
0
=&(\lambda\!+\!\mu \!+\!\nu\beta_{2,7})a_{2,7}
\\
&
-9\frac{(\beta_{3,12}-2)a_{3,12}+ (\beta_{7,2}-2)a_{7,2}  }{20\beta_{2,7}}\nonumber
+\frac{(\beta_{3,2}-2)a_{3,2}+ (\beta_{7,12}-2)a_{7,12}}{4\beta_{2,7}},\nonumber%\vspace{1.5mm}
\\
& ...,
%\end{array}\right\}\label{N5}
\end{align*}\vspace{-8mm}
\begin{align*}
0
=&(\lambda+\mu +\nu\beta_{n,m})a_{n,m}
\\
&+\frac{n-m}{4N\beta_{n,m}}[(\beta_{n-N,m-N}-2)a_{n-N,m-N}
-(\beta_{n+N,m+N}-2)a_{n+N,m+N}]\nonumber
\\&+\frac{n+m}{4N\beta_{n,m}}[(\beta_{n-N,m+N}-2)a_{n-N,m+N}
- (\beta_{n+N,m-N}-2)a_{n+N,m-N} ]\Big\},\nonumber
\\
&...
\end{align*}
for $(n,m)= (2+5j,3+5k),$ $ (3+5j,2+5k),$ $(2+5j,2+5k)$ and $(3+5j,3+5k)$, respectively,  with positive integers $j$ and $k$.

 For the  critical spectral solutions $(\i\omega_c,\nu_c,\mu_c, \psi_c)$  and $(-\i\omega_c,\nu_c,\mu_c, \bar \psi_c)$ at the pure viscosity case, we obtain
  $(\omega_c,\nu_c,\mu_c) \approx(0.1397, 0.2064,0)$. Their corresponding critical eigenfunction  $\psi_c$ generated %respectively by the modes $\sin\frac{x}3\sin\frac{2y}3$ for $N=3$ and
  by the modes $\sin\frac{2x}5 \sin\frac{3y}5$ and
  $\sin\frac{3x}5 \sin\frac{2y}5$  is    displayed in Figure \ref{f2}.

 \begin{figure}
 \centering
\includegraphics[height=.45\textwidth, width=.45\textwidth]{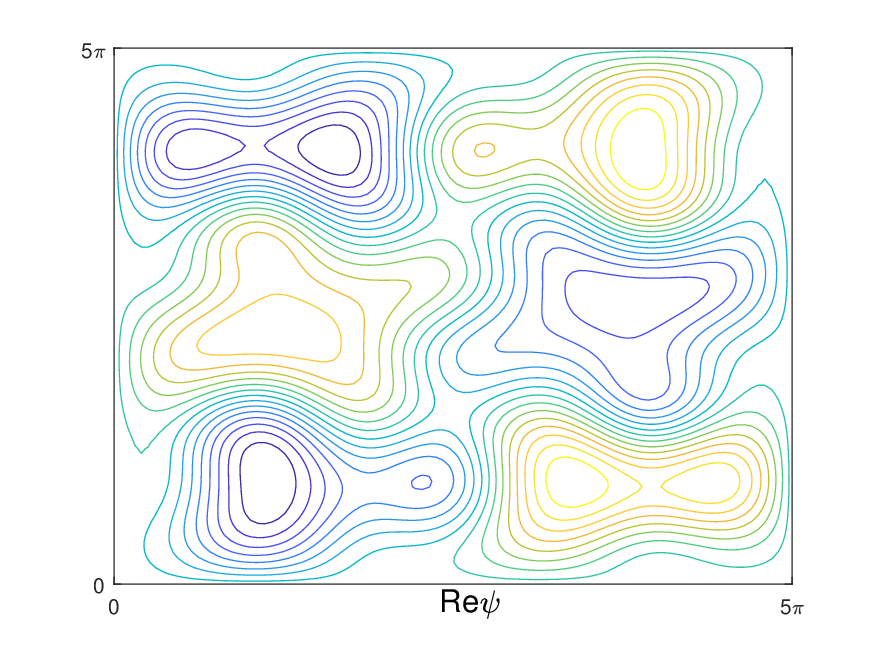}
\includegraphics[height=.45\textwidth, width=.45\textwidth]{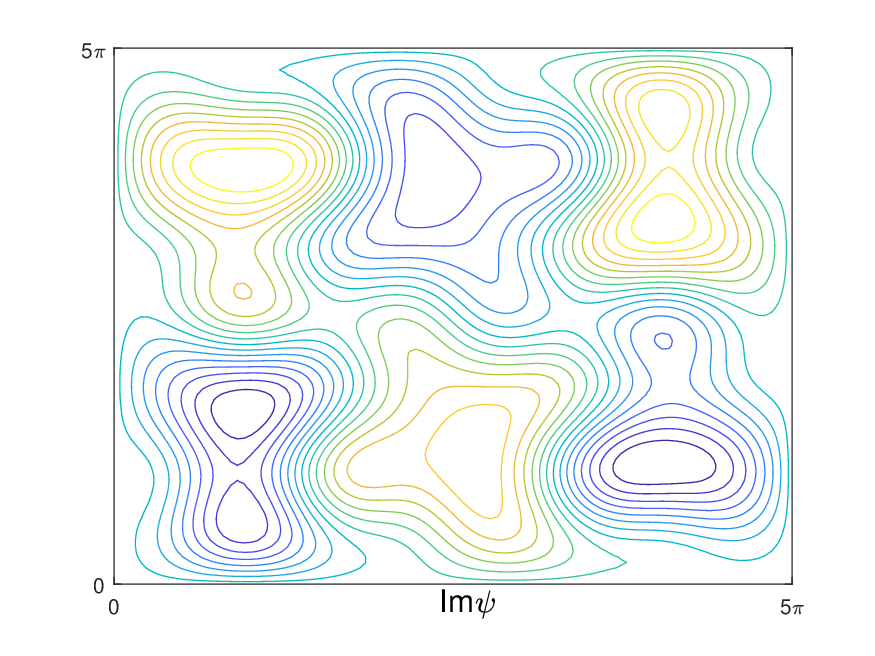}
 \caption{
  Real and imaginary parts of the critical eigenfunction $\psi_c\in E_5$ at $(\i\omega_c,\nu_c,\mu_c) \approx(0.1397, 0.2064,0)$ .}

  \label{f2}
 \end{figure}

It should be noted that  the number of algebraic equations are significantly reduced after the use of (\ref{anm}), as (\ref{nspp}) can be replaced by
\begin{align}
0
=&\sum_{n,m\ge -N; \, n\, even, \, m \, odd} \sin \frac{nx}{N}\sin \frac{my}{N}\Big\{ (\lambda+\mu +\nu\beta_{n,m})\i a_{m,n}
\\
&+\frac{n-m}{4N\beta_{n,m}}[(\beta_{n-N,m-N}-2)a_{n-N,m-N}
-(\beta_{n+N,m+N}-2)a_{n+N,m+N}]\nonumber
\\&+\frac{n+m}{4N\beta_{n,m}}[(\beta_{n-N,m+N}-2)a_{n-N,m+N}
- (\beta_{n+N,m-N}-2)a_{n+N,m-N} ]\Big\},\nonumber
\end{align}
which only involves the expansion coefficients with the  same parity  indices.

\section{ Proof of Theorem \ref{th1} }

By the derivation  of \cite[Lemmas 7.3-7.6]{JS}, we have the following invertibility lemma.
\begin{Lemma}\label{ll}
For a critical value $(\mu_c,\nu_c)$ with $\nu_c>0$ and $\mu_c \ge 0$, we assume that the critical spectral problem
\be(-\mu_c +\nu_c\Delta + \Delta ^{-1}L) \psi = \lambda \psi,\,\,\,\, \Im\lambda \ne 0 \mbox{ and } 0\ne \psi\in H^4\ee
has exactly two spectral solutions
$(-\i \omega_c, \psi_c)$ and $(\i \omega_c, \bar \psi_c)$.
Then the operator
$(-\omega_c\p_s-\mu + \nu_c\Delta +\Delta^{-1} L)$ is a bijection mapping $PC^{2+2\alpha, 1+\alpha}_{2\pi}$ onto $PC^{2\alpha, \alpha}_{2\pi}$ for $0<\alpha < \frac12 $
and
\be \|P\psi\|_{2+2\alpha,1+\alpha}\le C\|(-\omega_c\p_s+ \nu_c\Delta +\Delta^{-1} L)P\psi\|_{2\alpha,\alpha},
\ee
where $C$ is a generic constant throughout the paper.
\end{Lemma}

Theorem \ref{th1}  is implied from Theorem \ref{th2}  and the following theorem.

\begin{Theorem} \label{th3} In addition to the assumption of Lemma \ref{ll}, we assume that
\bbe (\psi_c, \Delta \psi_c^*)\ne 0\label{10a}\bee
and
\bbe \Re\frac{(-\mu\psi_c+\nu\Delta\psi_c, \Delta \psi_c^*)}{(\psi_c, \Delta \psi_c^*)}\ne 0.
\bee
Then  for any  $\epsilon \ne 0$ sufficiently small,  there exists a unique  element \be (\omega, \sigma, \phi) \in (0,\infty)\times (0, \infty)\times  PC^{2+2\alpha,1+\alpha}_{2\pi},\ee
 dependant smoothly on $\epsilon$,  in a neighborhood of $(\omega_c,1, 0)$  so that
\bbe
&&\psi=\epsilon (e^{-\i s} \psi_c+e^{\i s} \bar\psi_c + \phi),  \label{solu}
\\
&& (\mu, \nu) = \sigma (\mu_c, \nu_c)\nonumber
\bee  solves (\ref{oo2}) and (\ref{aabb}).
\end{Theorem}

Hopf bifurcation theorem under sufficient conditions  has been well studied. For example, Crandall and Rabinowitz \cite[Theorem 1.11]{C} and Kielh\"ofer \cite[Theorem I.8.2]{KK} considered    general nonlinear dynamic systems and provided sufficient conditions for the occurrence of Hopf bifurcation. Amongst these conditions, the principal sufficient ones are eigenfunction simplicity and eigenvalue transversal crossing conditions. Independently of \cite{JS}, Similar    Hopf bifurcation analysis on Navier-Stokes equations was also   given by Iudovich \cite{Iu1,Iu2} and Sattinger \cite{Sat}.
Nevertheless, the Hopf bifurcation analysis on a Navier-Stokes fluid motion model studied by Joseph and Sattinger \cite{JS} is most suitable to the present study.  However, strictly speaking, the present fluid motion equation and that of  \cite{JS} have several differences with respect to  boundary conditions,  energy dissipations and equation  formulations (vorticity formulation in the present and  velocity formulation in \cite{JS}). Thus for the completion of rigorous analysis, we prefer to  provide a complete proof  for  Theorem \ref{th3} by  developing the analysis  from Joseph and Sattinger \cite{JS} via the use of  the solution perturbation form (\ref{solu}). This perturbation can also  be found in Rabinowitz \cite{R} for a steady-state bifurcation
 problem. In addition, our proof is straightforward and is obtained simply from setting up a contraction mapping, which is an immediate consequence of the equation nonlinearity  and the eigenfunction simplicity and eigenvalue transversal crossing  conditions.

\begin{proof}
We adopt the notation
\be \varphi_c  \triangleq e^{-\i s} \psi_c,\,\,\,\, \bar \varphi_c \triangleq e^{\i s} \bar\psi_c,   \,\,\, \phi_c\triangleq \varphi_c+ \bar \varphi_c.
\ee

On the substation of $\psi =\epsilon (\phi_c+\phi)$ into (\ref{aabb}) and with division by $\epsilon$ to the resultant equation, we have
\bbe %\label{aabb}
0&=& \frac1\epsilon \left( (-\omega_c \p_s  -\mu_c+ \nu_c \Delta   +\Delta^{-1}L)\psi- (\omega-\omega_c) \p_s\psi \right.\nonumber
\\
&& \left.-(\mu-\mu_c)  \psi+ (\nu-\nu_c) \Delta\psi + \Delta^{-1}J(\psi,\Delta \psi)\right)\nonumber
\\
&=& (-\omega_c\p_s  -\mu_c +\nu_c\Delta  + \Delta^{-1}L)\phi -(\omega-\omega_c)\p_s (\phi_c+\phi)\label{aa1}
\\
&&+  (\sigma-1)(-\mu_c +\nu_c\Delta)(\phi_c+\phi)+ \epsilon\Delta^{-1} J(\phi_c+\phi, \Delta(\phi_c+\phi)).\nonumber
\bee
Since $P\phi=\phi$ and $\phi$ is $2\pi$ periodic, we have
\bbe \langle  \p_s\phi, \Delta \varphi_c^*\rangle =\frac1{2\pi}\int^{2\pi}_0 (\p_s \phi, e^{-\i s}\Delta\psi_c^*)ds=0,\label{c1}
\bee
\bbe
\langle  \p_s \phi_c, \Delta\varphi_c^*\rangle =\frac1{2\pi}\int^{2\pi}_0\!\!\!\! (\p_s ( e^{-\i s}\psi_c+ e^{\i s}\bar\psi_c), e^{-\i s}\Delta\psi_c^*)ds=-\i(\psi_c,\Delta\psi_c^*),\label{c2}
\bee
\begin{align}
\langle  -\mu_c\phi_c+\nu_c\Delta \phi_c, \Delta\varphi_c^*\rangle &=\frac1{2\pi}\int^{2\pi}_0\!\!\!\! ((-\mu_c+\nu_c\Delta ) ( e^{-\i s}\psi_c+ e^{\i s}\bar\psi_c), e^{-\i s}\Delta\psi_c^*)ds
\nonumber
\\
&=(-\mu_c\psi_c+\nu_c\Delta \psi_c,\Delta\psi_c^*).\label{c3}
\end{align}
Therefore, taking the inner product of  (\ref{aa1}) with $\Delta \varphi_c^*$, we have
\begin{align*} 0&= \i (\omega-\omega_c)(\psi_c,\Delta\psi_c^*)+(\sigma-1)(-\mu_c\psi_c+\nu_c\Delta\psi_c,\Delta \psi_c^*)
\\
&+ (\sigma-1)\langle -\mu_c\phi+\nu_c\Delta\phi,\Delta \varphi_c^*\rangle + \epsilon \langle \Delta ^{-1}J(\phi_c +\phi,\Delta (\phi_c+\phi)),\Delta\varphi_c^*\rangle .
\end{align*}
That is, after the division of the previous equation by $(\psi_c,\Delta\psi_c^*)$ due to   (\ref{10a}),
\begin{align} 0&= \i (\omega-\omega_c)+(\sigma-1)\frac{(-\mu_c\psi_c+\nu_c\Delta\psi_c,\Delta \psi_c^*)}{(\psi_c,\Delta\psi_c^*)}\label{aa101}
\\
&+ (\sigma-1)\frac{\langle -\mu_c\phi+\nu_c\Delta\phi,\Delta \varphi_c^*\rangle }{(\psi_c,\Delta\psi_c^*)}+ \epsilon \frac{\langle \Delta ^{-1}J(\phi_c +\phi,\Delta (\phi_c+\phi)),\Delta \varphi_c^*\rangle }{(\psi_c,\Delta\psi_c^*)}.\nonumber
\end{align}
Separating the  real and imaginary parts of (\ref{aa101}), we have
\bbe
\sigma-1=f_\epsilon
\bee
with
\bbe
f_\epsilon\triangleq  -\frac{(\sigma\!-\!1)\Re\displaystyle\frac{\langle -\mu_c\phi\!+\!\nu_c\Delta\phi,\Delta \psi_c^*\rangle }{(\psi_c,\Delta\psi_c^*)}\!+\!\epsilon \Re\frac{\langle \Delta ^{-1}J(\phi_c\!+\!\phi,\Delta (\phi_c\!+\!\phi)),\Delta \varphi_c^*\rangle }{(\psi_c,\Delta\psi_c^*)}}
{\Re\displaystyle\frac{(-\mu_c+\nu_c\Delta\psi_c,\Delta \psi_c^*)}{(\psi_c,\Delta\psi_c^*)}},\label{d1d}
\bee
\begin{align} \omega-\omega_c&=-f_\epsilon\Im\frac{(-\mu_c\psi_c+\nu_c\Delta\psi_c,\Delta \psi_c^*)}{(\psi_c,\Delta\psi_c^*)}- (\sigma-1)\Im\frac{\langle -\mu_c\phi+\nu_c\Delta\phi,\Delta \varphi_c^*\rangle }{(\psi_c,\Delta\psi_c^*)}\nonumber
\\
&- \epsilon\Im \frac{\langle \Delta ^{-1}J(\phi_c +\phi,\Delta (\phi_c+\phi)),\Delta \varphi_c^*\rangle }{(\psi_c,\Delta\psi_c^*)}\triangleq g_\epsilon.\label{aa100}
\end{align}
By using the mappings $f_\epsilon$ and $g_\epsilon$, we rewrite
(\ref{aa1}) as
\bbe\lefteqn{ -(-\omega_c\p_s -\mu_c +\nu_c\Delta  \!+ \!\Delta^{-1}L)\phi}\nonumber
\\&=& -(\omega\!-\!\omega_c)\p_s \phi+ (\sigma-1)(-\mu_c+\nu_c\Delta)\phi
\!-\!g_\epsilon \p_s \phi_c\!\nonumber
\\
&&+ \!f_\epsilon (-\mu_c+\nu_c\Delta)\phi_c\!+ \! \epsilon\Delta^{-1} J(\phi_c\!+ \!\phi, \Delta(\phi_c\!+ \!\phi)).\label{bb2}
\bee
Let $\Phi_\epsilon$ denote the right-hand side of (\ref{bb2}). To check $\Phi_\epsilon$  to be orthogonal to $\Delta \varphi_c^*$ and $\Delta \overline{\varphi_c^*}$, we see that , by (\ref{c1})-(\ref{c3}),
\bbe
\langle \Phi_\epsilon,\Delta \varphi_c^*\rangle
&=& (\sigma-1)\langle -\mu_c \phi+\nu_c\Delta\phi,\Delta \varphi_c^*\rangle +\i g_\epsilon  (\psi_c,\Delta \psi_c^*)\nonumber
\\&&+ f_\epsilon (-\mu_c\psi_c+\nu_c\Delta\psi_c,\Delta \psi_c^*)\nonumber
\\
&&+\epsilon\langle \Delta^{-1} J(\phi_c\!+ \!\phi, \Delta(\phi_c\!+ \!\phi)), \Delta \varphi_c^*\rangle ,\label{bb1}
\bee
or
\bbe
\frac{\langle \Phi_\epsilon,\Delta \varphi_c^*\rangle }{( \psi_c,\Delta \psi_c^*)}
&=& (\sigma-1)\frac{\langle \mu_c\phi+\nu_c\Delta\phi,\Delta \varphi_c^*\rangle }{( \psi_c,\Delta \psi_c^*)}+\i g_\epsilon\nonumber
\\
&& + f_\epsilon \frac{(-\mu_c\psi_c+\nu_c\Delta\psi_c,\Delta \psi_c^*)}{( \psi_c,\Delta \psi_c^*)}\nonumber
\\
&&+\epsilon\frac{\langle \Delta^{-1} J(\phi_c\!+ \!\phi, \Delta(\phi_c\!+ \!\phi)), \Delta \varphi_c^*\rangle }{{( \psi_c,\Delta \psi_c^*)}}.\label{bbb1}
\bee
This together with (\ref{d1d}) and (\ref{aa100})  gives
\be \langle \Phi_\epsilon,\Delta \varphi_c^*\rangle =0. \ee
Arguing in the same way, we have
\be \langle \Phi_\epsilon,\Delta \overline{\varphi_c^*}\rangle =0. \ee
Hence we have $P \Phi_\epsilon =\Phi_\epsilon$. Applying Lemma \ref{ll} to  (\ref{aa1}), we have
\bbe \phi &=& (-\omega_c\p_s -\mu_c +\nu_c\Delta  + \Delta^{-1}L)^{-1} \Phi_\epsilon\triangleq h_\epsilon.  \label{1q}
\bee

Thus (\ref{aa1}) is rewritten into the fixed point problem
\be
(\sigma-1,\omega-\omega_c,\phi ) &=& F_\epsilon (\sigma-1,\omega-\omega_c,\phi)
\\ &\triangleq& \Big(f_\epsilon(\sigma-1,\phi),g_\epsilon(\sigma-1,\phi),h_\epsilon(\sigma-1,\omega-\omega_\epsilon,\phi)\Big).
\ee

Now it remains to check that $F_\epsilon $ is a contraction mapping on the complete matric space
\be X &=& \Big\{ (\sigma-1,\omega-\omega_c,\phi )\in (-\infty,\infty) \times (-\infty,\infty) \times PC^{2+2\alpha, 1+\alpha}_{2\pi} \Big|
\\
&&\,\, \|(\sigma-1,\omega-\omega_c,\phi )\|_{X} = |\sigma-1|+|\omega-\omega_c|+ \|\phi\|_{2+2\alpha,1+\alpha} \le M \Big\}
\ee
for a constant $M$ to be determined.

To check the injection property $F_\epsilon: X \mapsto X$, we use the H\"older estimate and the boundedness of the operator $\Delta^{-1}\nabla$ to produce
%\be \|F_\epsilon\|_{X} &=& |f_\epsilon|+ |g_\epsilon|+ \| h_\epsilon \|_{2+2\alpha,1+\alpha}
 %\ee
\begin{align*} |f_\epsilon| =& \left|\frac{(\sigma\!-\!1)\Re\displaystyle\frac{\langle -\mu_c\phi\!+\!\nu_c\Delta\phi,\Delta \psi_c^*\rangle }{(\psi_c,\Delta\psi_c^*)}\!+\!\epsilon \Re\frac{\langle \Delta ^{-1}J(\phi_c\!+\!\phi,\Delta (\phi_c\!+\!\phi)),\Delta \varphi_c^*\rangle }{(\psi_c,\Delta\psi_c^*)}}
{\Re\displaystyle\frac{(-\mu_c+\nu_c\Delta\psi_c,\Delta \psi_c^*)}{(\psi_c,\Delta\psi_c^*)}}\right|
\\
\leq &C\Big( |\sigma-1 | \|\Delta \phi\|_{L_2} +\epsilon \|\Delta ^{-1}J(\phi_c,\Delta \phi_c)\|_{L_2}
\\
&+\epsilon\|\nabla \phi\|_{L_2}+ \epsilon\|\Delta \phi\|_{L_2}+  \epsilon\|\nabla \phi\|_{L_\infty}\|\Delta \phi\|_{L_2}
\Big)
\\
\leq& C(M^2+ \epsilon + \epsilon M + \epsilon M^2),
\end{align*}

\begin{align*}
|g_\epsilon| &\le
\Bigg|f_\epsilon\Im\frac{(-\mu_c\psi_c+\nu_c\Delta\psi_c,\Delta \psi_c^*)}{(\psi_c,\Delta\psi_c^*)}+ (\sigma-1)\Im\frac{\langle -\mu_c\phi+\nu_c\Delta\phi,\Delta \varphi_c^*\rangle }{(\psi_c,\Delta\psi_c^*)}%\label{aa1000}
\\
&+ \epsilon\Im \frac{\langle \Delta ^{-1}J(\phi_c +\phi,\Delta (\phi_c+\phi)),\Delta \varphi_c^*\rangle }{(\psi_c,\Delta\psi_c^*)}\Bigg|
\\
&\leq C ( |f_\epsilon| + |\sigma-1 | \|\Delta \phi\|_{L_2} + \epsilon \|\Delta ^{-1}J(\phi_c,\Delta \phi_c)\|_{L_2}
\\&+ \epsilon\|\nabla \phi\|_{L_2}+ \epsilon\|\Delta \phi\|_{L_2}+  \epsilon\|\nabla \phi\|_{L_\infty}\|\Delta \phi\|_{L_2})
\\
&\leq C( M^2+\epsilon + \epsilon M + \epsilon M^2),
\end{align*}
and, by Lemma \ref{ll},
\be \| h_\epsilon\|_{2+2\alpha,1+\alpha} &\le & C \|\Phi_\epsilon\|_{2\alpha,\alpha}
\\
&\leq & C \|(\omega-\omega_c)\p_s \phi+(\sigma-1)(-\mu_c\phi+\nu_c\Delta\phi)
+g_\epsilon\p_s\phi_c
\\ &&+ f_\epsilon(-\mu_c\phi_c+\nu_c\Delta\phi_c)
+ \epsilon \Delta^{-1} J(\phi_c+\phi, \Delta (\phi_c+\phi))\|_{2\alpha,\alpha}
\\
&\leq & C \Big( |\omega-\omega_c|\,\| \phi\|_{2+2\alpha,1+\alpha} +|\sigma-1|\,\| \phi\|_{2+2\alpha,1+\alpha}
+|g_\epsilon|
+ |f_\epsilon|
\\ &&
+ \epsilon \|\Delta^{-1} J(\phi_c, \Delta \phi_c)\|_{2\alpha,\alpha}+ \epsilon\| \nabla\phi\|_{2\alpha,\alpha} +\epsilon\| \Delta \phi\|_{2\alpha,\alpha}
\\&&
+\epsilon\| \nabla\phi\|_{2\alpha,\alpha} \| \Delta \phi\|_{2\alpha,\alpha} \Big)
\\
&\le& C(M^2+|g_\epsilon|+ |f_\epsilon| + \epsilon +\epsilon M +\epsilon M^2)
\\
&\leq& C(M^2+ \epsilon + \epsilon M + \epsilon M^2).
\ee
Collecting terms
, we have
\bbe \|F_\epsilon\|_{X} &=& |f_\epsilon|+ |g_\epsilon|+ \| h_\epsilon \|_{2+2\alpha,1+\alpha}
 \le C(M^2+ \epsilon + \epsilon M + \epsilon M^2). \label{Xz}
\bee
For the constant $C$ in  (\ref{Xz}), we  take  the parameters $\epsilon$  and $M$ sufficiently small so that
\be  2 C\epsilon \le M \,\,\mbox{ and }
C(M +\epsilon  + \epsilon M) \leq \frac12.
\ee
This gives the injection property
\be
\|F_\epsilon (\omega-\omega_c,\sigma-1,\phi)\|_{X} \le M \,\,\, \mbox{ for } (\omega-\omega_c,\sigma-1,\phi)\in X.
\ee
Similarly, we have the contraction property
\be
\lefteqn{\|F_\epsilon(\omega'-\omega_c,\sigma'-1,\phi')- F_\epsilon(\omega''-\omega_c,\sigma''-1,\phi'')\|_{X} }
\\
&\le& \frac12\|(\omega'-\omega_c,\sigma'-1,\phi')- (\omega''-\omega_c,\sigma''-1,\phi'')\|_{X}
\ee
for any $(\omega'-\omega_c,\sigma'-1,\phi'),  (\omega''-\omega_c,\sigma''-1,\phi'')\in X$, provided that $\epsilon$ and $M$ are  sufficiently small.

Therefore, by the Banach contraction mapping principle,  the operator $F_\epsilon$ has a unique fixed point in $X$. We thus have the  desired result.

The proof of Theorem \ref{th3} is complete.
\end{proof}

\

\noindent \textbf{Acknowledgement.} This research is supported by The Shenzhen Natural Science Fund of China (the stable Support Plan Program No. 20220805175116001).
%\noindent \textbf{Availability of data.} No new data were created or analysed in this study.
% Data available on request from the author.

 \

\end{document}